\documentclass[11pt]{article}

\usepackage{amssymb,amsmath,amsthm,mathtools, natbib, booktabs, float}
\usepackage{array,geometry}  
\usepackage{setspace}
\usepackage{bm,paralist,dsfont,nicefrac,bbm}
\usepackage{etoolbox}
\usepackage[driverfallback=hypertex,pagebackref=true,colorlinks]{hyperref}
\hypersetup{linkcolor=[rgb]{.7,0,0}}
\hypersetup{citecolor=[rgb]{0,.7,0}}
\hypersetup{urlcolor=[rgb]{.7,0,.7}}
\renewcommand*{\backref}[1]{}
\renewcommand*{\backrefalt}[4]{%
    \ifcase #1 (Not cited.)%
    \or        (Cited on page~#2)%
    \else      (Cited on pages~#2)%
    \fi}

\makeatletter
\newcommand{\pushright}[1]{\ifmeasuring@#1\else\omit\hfill$\displaystyle#1$\fi\ignorespaces}
\newcommand{\pushleft}[1]{\ifmeasuring@#1\else\omit$\displaystyle#1$\hfill\fi\ignorespaces}
\makeatother

\usepackage{xcolor}

\newcommand{\BibTeX}{\rm B\kern-.05em{\sc i\kern-.025em b}\kern-.08em\TeX}

\usepackage{array, enumitem, mathtools, tabularx, hyperref}

\usepackage{thmtools}
\usepackage{thm-restate}

\makeatletter
\@ifundefined{newcounteralias}{}{%
  \renewcommand\thmt@autorefsetup{%
    \@xa\def\csname\thmt@envname autorefname\@xa\endcsname%
    \@xa{\thmt@thmname}%
  }%
}
\makeatother

\usepackage[capitalize, nameinlink, noabbrev]{cleveref}
	\crefalias{section}{section}%
	\crefalias{subsection}{section}

\Crefname{lemma}{Lemma}{Lemmas}
\newtheorem{theorem}{Theorem}[section]
\newtheorem{lemma}{Lemma}[section]
\newtheorem{corollary}{Corollary}[section]

\newcommand{\A}{\mathcal A}
\newcommand{\BoBW}{\mathrm{BoBW}}

\newcommand{\eps}{\varepsilon}

\newcommand{\EFone}{\mathrm{EF1}}
\newcommand{\EFX}{\mathrm{EFX}}
\newcommand{\EFable}{\mathrm{EFable}}
\newcommand{\fPO}{\mathrm{fPO}}
\newcommand{\F}{\mathcal F}

\newcommand{\one}{\mathbf 1}
\newcommand{\PO}{\mathrm{PO}}
\newcommand{\R}{\mathbb R}
\newcommand{\Rpp}{\R_{>0}}

\newcommand{\supp}{\operatorname{supp}}

\newcolumntype{Y}{>{\raggedright\arraybackslash}X}

\AtBeginDocument{%
}

\title{Best-of-Both-Worlds Fair Division via Subsidies}
\author{
\begin{tabular}{cc}
	\textbf{Surya Panchapakesan} & \textbf{Rohit Vaish} \\[2pt]
    {\small National University of Singapore} & {\small IIT Delhi} \\[2pt] \href{mailto:csurya.panchapakesan@u.nus.edu}{\small\texttt{csurya.panchapakesan@u.nus.edu}}
    &
    \href{mailto:rvaish@iitd.ac.in}{\small\texttt{rvaish@iitd.ac.in}}
\end{tabular}
}

\date{}

\begin{document}
\maketitle 

\begin{abstract}
We study the relationship between fair division with \emph{subsidies} and \emph{best-of-both-worlds} (BoBW) fairness for indivisible items. We prove a general reduction of BoBW existence to envy-freeability. Specifically, for any scale-free ex-post property $P$ (e.g., EF1, EFX, PO) and any scale-free valuation class $\mathcal{V}$ (e.g, additive, submodular, subadditive, monotone), if every instance in $\mathcal{V}$ admits a $P$-allocation that can be made envy-free using subsidies, then every such instance admits an ex-ante envy-free lottery supported on $P$-allocations. The proof uses linear programming duality. As applications, the reduction gives alternative proofs of BoBW existence results for additive goods, chores, and mixtures of goods and chores. The converse of our general reduction holds for two agents but fails for three: We provide a four-good counterexample with monotone submodular valuations where an ex-ante envy-free and ex-post EF1 lottery exists, but no allocation is simultaneously EF1 and envy-freeable. For lexicographic valuations, we prove simultaneous EFX, Pareto optimality, and envy-freeability (and hence BoBW guarantees for EFX+PO) for three agents, and show that even EFX and envy-freeability can be incompatible for four agents.
\end{abstract}

\section{Introduction}
\label{sec:Introduction}

Envy-freeness is a central solution concept in resource allocation. It requires each agent to weakly prefer its own bundle to every other agent's bundle~\citep{F67resource}. Often the ``gold standard'' of fairness, envy-freeness enjoys a strong existential and computational support when the resources are infinitely divisible~\citep{V74equity}. However, indivisibility makes this requirement impossible even when two agents share a single valuable good. One can therefore consider relaxations, such as envy-freeness up to one item (EF1) or up to any item (EFX).

Such relaxations guarantee fairness despite indivisibility, but they may still leave some agents dissatisfied. Indeed, when one good is to be divided between two agents, the agent who does not receive the good might consider the allocation unfair, even though such an allocation satisfies approximate fairness in the sense of EF1 or EFX. To remedy this, the literature has explored two prominent approaches: fairness through \emph{subsidies} and fairness through \emph{randomization}.

In fair division with subsidies, envious agents can be compensated through monetary payments. An allocation is called \emph{envy-freeable} if nonnegative payments make its bundle–payment pairs envy-free~\citep{HS19fair}. On the other hand, fairness through randomization entails constructing a probability distribution, or a lottery, over deterministic allocations. Such a lottery must be envy-free \emph{ex ante}~(i.e., in expected utility), while every realized allocation satisfies an approximate fairness condition \emph{ex post}, such as envy-freeness up to one good (EF1). Combining ex-ante and ex-post fairness properties is also known as best-of-both-worlds (BoBW) fairness~\citep{AFS+24best}. In particular, the lottery does not use any payments.

The subsidy and BoBW literatures have developed largely independently of each other, yet several of their existence guarantees are strikingly parallel. For additive goods, for example, balanced EF1 allocations can be made envy-free with subsidies~\citep{BDN+20one}, and balanced EF1 outcomes can be randomized to eliminate
envy in expectation~\citep{AFS+24best}. It is interesting to ask whether these two seemingly different results are fundamentally connected.

Our main result reduces the existence of a BoBW lottery to the existence of suitable subsidy-supported allocations. Let $P$ be any property that is invariant under multiplying each agent's valuation by a positive constant; examples include EF1, EFX, balancedness, Pareto optimality, and conjunctions of these properties. We show that if every positive rescaling of a profile admits a $P$-allocation that is envy-freeable, then the original profile admits an ex-ante envy-free lottery all of whose outcomes satisfy $P$~(\Cref{thm:bobw-to-subsidy}).

The reduction has both positive and negative consequences (\Cref{sec:consequences}). It gives alternative existence proofs for additive goods, chores, and mixed manna, and it extends to weighted entitlements. For two agents, it has an exact converse: For every positive rescaling, any ex-ante envy-free lottery has an envy-freeable allocation in its support at that scale. The contrapositive of our reduction shows that failure of a BoBW guarantee implies the incompatibility of $P$ and envy-freeability at some positive rescaling~(\Cref{cor:impossibility}).

\begin{table}[!t]
\centering
\small
\setstretch{1}
\setlength{\tabcolsep}{5pt}
\renewcommand{\arraystretch}{1.08}
\begin{tabularx}{0.8\linewidth}{@{}
  >{\raggedright\arraybackslash}p{0.225\linewidth}
  >{\raggedright\arraybackslash}X
  >{\raggedright\arraybackslash}p{0.16\linewidth}@{}}
\toprule
\textbf{Setting} & \textbf{Main consequence} & \textbf{Reference} \\
\midrule
\multicolumn{3}{@{}l}{\textit{New guarantees and separations}} \\
\addlinespace[4pt]
Lexicographic goods\newline $n=3$
& $\EFX+\PO+\EFable$ always exists;\newline hence $\BoBW(\EFX+\PO)$.
& Theorem~\ref{thm:lex-three} \\
\addlinespace[4pt]
Additive goods\newline $n=2$
& $\EFX_{+}+\PO+\EFable$ always exists.
& Corollary~\ref{cor:two-agent corollary} \\
\addlinespace[4pt]
Monotone submodular\newline goods; $n=3$
& $\EFone+\EFable$ can fail even when\newline $\BoBW(\EFone)$ exists.
& Theorem~\ref{thm:counterexample-submodular-valuations} \\
\addlinespace[4pt]
Lexicographic goods\newline $n=4$
& $\EFX+\EFable$ can fail even when\newline $\BoBW(\EFX+\PO)$ exists.
& Theorem~\ref{thm:counterexample-lexicographic-four} \\
\addlinespace[4pt]
Additive goods
& $\EFone+\fPO+\EFable$ can fail.
& Corollary~\ref{impossible:EF1+fPO+Efable} \\
\addlinespace[4pt]
Additive chores\newline Unequal entitlements
& $\mathrm{WEF1}+\mathrm{WEFable}$ can fail.
& Corollary~\ref{impossible: WEF1 and EFable chores} \\
\addlinespace[5pt]
\midrule
\multicolumn{3}{@{}l}{\textit{Alternative proofs of known results}} \\
\addlinespace[4pt]
Additive goods\newline or chores
& $\BoBW(\text{balanced EF1})$ always exists.
& Corollary~\ref{cor:additive-bobw} \\
\addlinespace[4pt]
Additive mixed manna
& $\BoBW(\EFone)$ always exists.
& Corollary~\ref{cor:additive-bobw} \\
\addlinespace[4pt]
Additive goods\newline Unequal entitlements
& $\mathrm{WEF1}+\mathrm{WEFable}$ can fail.
& Corollary~\ref{impossible: WEF1 and EFable goods} \\
\bottomrule
\end{tabularx}
\caption{Selected consequences of the subsidy--BoBW connection. $\BoBW(P)$ denotes an ex-ante EF lottery satisfying $P$ ex post.
EFable and WEFable denote ordinary and weighted envy-freeability.}
\label{tab:bobw-subsidy-domains}
\end{table}

Our reduction connects several existing results in the subsidy and BoBW literatures and contributes new observations for both problems. Some of the key consequences of our reduction are as follows~(also see \Cref{tab:bobw-subsidy-domains}):

\begin{itemize}[topsep=0.2cm]
    \item For additive goods, we recover the existence result of~\citet{AFS+24best} on ex-ante EF and ex-post EF1 lotteries using the existence of an EF1 and envy-freeable allocation~\citep{BDN+20one}. We also recover their BoBW existence for additive chores, as EF1 and envy-freeability are known to be compatible in that setting~\citep{WZZ23}.
    \item For mixtures of additive goods and chores, we recover the BoBW existence of~\citet{ABLMSTW26mixed} using the result of~\citet{ALMS25mixed} on the existence of EF1 and envy-freeable allocations.
    \item The compatibility of EF1 and envy-freeability for monotone goods was left open by~\citet{ALMS25mixed}. We answer this question in the negative with a three-agent, four-good example with submodular valuations~(\Cref{thm:counterexample-submodular-valuations}). The instance nevertheless admits a BoBW lottery, showing that the converse of our reduction does not hold for three agents.  
    
    \item For lexicographic goods,~\citet{KPVVY25fairer} obtained ex-ante $\frac{9}{10}$-EF alongside ex-post EFX and Pareto optimality (PO). We prove that for three agents, an ex-ante EF and ex-post EFX+PO lottery exists by showing that an EFX, PO, and envy-freeable allocation always exists~(\Cref{thm:lex-three}). 
    
    \item Our reduction extends to the setting where agents have unequal entitlements (or weights). As a result, we obtain that for additive chores, weighted EF1 and weighted envy-freeability are incompatible by leveraging a BoBW impossibility result of~\citet{LWZ26binary}.    
\end{itemize}

\paragraph{Technical overview of our reduction.}
Our proof uses linear programming (LP) duality. The primal LP maximizes the smallest expected negative envy over lotteries supported on $P$-allocations. If no ex-ante envy-free lottery exists, the dual LP gives nonnegative weights on pairwise envy comparisons such that every $P$-allocation has strictly positive weighted total envy. Thus, the same weights rule out every admissible allocation. Adding a small perturbation makes the weights strictly positive. We then use a standard stationary-distribution argument to factor these weights into positive agent-specific scales and a \emph{circulation}---a flow with equal incoming and outgoing weights at every agent. At the resulting scaled profile, summing the subsidy inequalities for an envy-freeable $P$-allocation against the circulation cancels all payment terms. This contradicts the dual separation and establishes the desired lottery.

\paragraph{The power and limits of our reduction.}
Our reduction applies to arbitrary real-valued valuations and requires neither additivity nor monotonicity. The scale-freeness requirements accommodate a wide range of fairness and efficiency properties, including EF1, EFX, PO, and their conjunctions, as well as the additive, submodular, subadditive, and monotone valuation domains. The argument is also quite simple: it uses only standard LP duality and a stationary-distribution construction. %
Additionally, it guarantees a BoBW lottery supported only on $\mathcal{O}(n^2)$ allocations, independent of the number of items.

The reduction is, however, existential, and does not, by itself, provide an efficient algorithm for constructing the lottery. Moreover, the subsidy guarantee must hold after \emph{every} positive agent-specific rescaling. In particular, it is not enough to have an envy-freeable $P$-allocation at the original profile alone.

\subsection{Related Work}

\paragraph{Fair division with subsidy.}
\citet{HS19fair} studied envy-freeable allocations of indivisible goods and established bounds on the subsidies required to achieve envy-freeness. For additive valuations, they characterized envy-freeability by the absence of positive-weight directed envy cycles and by welfare optimality under bundle permutations.~\citet{A21achieving} noted that this characterization extends to arbitrary real-valued valuations; we recall these results in~\Cref{thm:efable-char}.~\citet{BDN+20one} provided an algorithm that returns a balanced, EF1 and envy-freeable allocation for additive goods, with at most one unit of subsidy per agent when every item value is at most $1$.~\citet{ALMS25mixed} extended simultaneous EF1 and envy-freeability to additive mixed goods and chores. 

Several works have studied subsidy bounds beyond additivity~\citep{KMS+25towards,lF25discrepancy,DS26subquadratic}. \citet{DS26subquadratic} proved that a total subsidy of $\mathcal{O}(n^{3/2}\sqrt{\log n})$ suffices when all valuations are nonnegative on every bundle, or all are nonpositive, provided single-item marginal values lie in $[-1,1]$. Notably, their result does not require monotonicity. However, these subsidy bounds do not automatically guarantee that the underlying allocation is EF1. Our submodular counterexample (\Cref{thm:counterexample-submodular-valuations}) shows that, beyond additive valuations, EF1 and envy-freeability can be incompatible.%

\citet{AHKS+25Subsidy} and \citet{EGE25Subsidy} studied the subsidy problem for agents with unequal entitlements. These works characterized weighted envy-freeability and provided lower and upper bounds on the subsidy required for various valuation classes. \citet{AHKS+25Subsidy} also showed that weighted envy-freeability is incompatible with a weakening of weighted EF1.

\paragraph{Best-of-both-worlds fairness.}
\citet{AFS+24best} established the existence of ex-ante envy-free and ex-post EF1 lotteries for additive goods. \citet{ABLMSTW26mixed} extended this guarantee to additive mixed goods and chores. Our reduction gives alternative existence proofs of these results; however, we do not provide an efficient algorithm. 

Beyond additivity, \citet{HSV24entitlements} obtained ex-ante EF and ex-post EF1 for multi-demand valuations, while \citet{FMNP24breaking} obtained ex-ante $1/2$-EF with ex-post EF1 and $1/2$-EFX for subadditive goods. For matroid-rank valuations, \citet{BEF21dichotomous} obtained ex-ante envy-freeness and ex-post Lorenz dominance through random priorities. 
\citet{KPVVY25fairer} obtained ex-ante $1/2$-EF with ex-post EFX-with-charity for monotone valuations, as well as a bounded-charity guarantee with ex-ante half-proportionality for subadditive valuations. These partial-allocation guarantees are distinct from the complete allocations considered here. Recent work by~\citet{F26subadditivebobw} studies BoBW guarantees for the maximum expectation share (MES) ex ante and the maximin share (MMS) ex post for subadditive valuations, rather than pairwise envy-freeness.

For stronger ex-post fairness,~\citet{HSV+21fair} studied EFX and PO under lexicographic preferences, and~\citet{KPVVY25fairer} obtained ex-ante $\frac{9}{10}$-EF alongside ex-post EFX+PO on that domain. For additive goods,~\citet{Aziz26po} proved that exact ex-ante EF is achievable alongside ex-post EFX$_{+}$+PO for two agents but fails for three agents. However, under lexicographic preferences, the above guarantee holds for three agents (\Cref{thm:lex-three}).

For the setting of unequal entitlements,~\citet{AGM23BoBW} showed that ex-ante weighted EF (WEF) and ex-post weighted EF1 (WEF1) are incompatible under additive valuations.
They also showed the existence of a lottery that is ex-ante WEF and ex-post weighted EF1 up to a transfer (WEF1T). Subsequently, ~\citet{LLTZ26weightedmixed} extended this positive result to additive mixed manna.

\section{Preliminaries}
\label{sec:Preliminaries}

For any nonnegative integer $r$, let $[r] \coloneqq \{1,2,\dots,r\}$, with $[0] = \emptyset$. 

Let $N = [n]$, where $n \geq 2$, be a set of agents and let $M$ be a finite set of indivisible items. A complete allocation $A = (A_1, \dots, A_n)$ is an ordered partition of $M$ with possibly empty bundles. Agent $i$ has a valuation function $v_i: 2^M \to \R$ that specifies its numerical value, or \emph{utility}, $v_i(S)$ for each subset, or \emph{bundle}, $S \subseteq M$. We assume $v_i(\emptyset) = 0$ for all agents $i \in N$.

\paragraph{Valuation domains.}
We abbreviate singleton values $v_i(\{g\})$ as $v_i(g)$. A valuation is \emph{additive} if for every subset $S \subseteq M$, $v_i(S) = \sum_{g\in S}v_i(g)$. Additive goods have nonnegative item values, additive chores have nonpositive item values, and mixed manna allows an item to have different signs for different agents. A valuation is \emph{monotone} if $v_i(S) \leq v_i(T)$ whenever $S \subseteq T$, \emph{submodular} if $v_i(S) + v_i(T) \geq v_i(S \cup T) + v_i(S \cap T)$, and \emph{subadditive} if $v_i(S \cup T) \leq v_i(S) + v_i(T)$.

A valuation is \emph{lexicographic} if it is additive and admits an order $g_1 \succ_i \cdots \succ_i g_m$ such that $v_i(g_r) > \sum_{s > r} v_i(g_s)$ for all $r \in [m]$. Equivalently, the highest-ranked good in the symmetric difference of two distinct bundles determines their comparison. 

A \emph{matroid-rank} valuation is a normalized monotone submodular valuation whose marginal increments are in $\{0,1\}$.

\paragraph{Fairness notions.}
An allocation $A$ is \emph{envy-free} (EF) if $v_i(A_i) \geq v_i(A_j)$ for all $i,j \in N$. For monotone goods, an allocation $A$ is \emph{envy-free up to one good} (EF1) if, for every $i,j \in N$, either $i$ does not envy $j$ or some $g \in A_j$ satisfies $v_i(A_i) \geq v_i(A_j \setminus \{g\})$. It is \emph{envy-free up to any good} (EFX) if this inequality holds for every $g \in A_j$ and every $i,j \in N$. For additive goods, EFX$_{+}$ is a weakening of EFX where only positively-valued items are removed from the envied bundle. For additive mixed manna, EF1 allows removing one item from the union of the two bundles; that is, either $i$ does not envy $j$ or some $g \in A_i \cup A_j$ satisfies $v_i(A_i \setminus \{g\}) \geq v_i(A_j \setminus \{g\})$.

\paragraph{Efficiency notions.}
An allocation is \emph{Pareto optimal} (PO) if no other complete allocation weakly improves every agent's utility and strictly improves at least one. For additive valuations, the notion of \emph{fractional Pareto optimality} (fPO) also excludes improvements by fractional assignments $x_{ig}\ge0$, $\sum_i x_{ig}=1$, with utility $\sum_g v_i(g)x_{ig}$. An allocation is \emph{balanced} if bundle sizes differ by at most one. We use fractional Pareto optimality only in the context of additive valuations.

\paragraph{Allocation properties and lotteries.}
Write $v=(v_1,\dots,v_n)$ for a valuation profile and let $\A$ denote the finite set of all complete allocations. For any property $P$, let $\F_P(v) \subseteq \A$ denote the set of allocations satisfying $P$ at profile $v$. A \emph{lottery} is a vector $q = (q_A)_{A \in \A}$ with $q_A \geq 0$ and $\sum_{A \in \A} q_A = 1$, and its support is $\supp(q) \coloneqq \{A \in \A : q_A > 0\}$. The lottery satisfies $P$ \emph{ex post} if $\supp(q) \subseteq \F_P(v)$. For $\alpha \in (0,1]$. it is \emph{ex-ante $\alpha$-envy-free} if, for all $i,j\in N$,
$$
  \sum_{A \in \A} q_A v_i(A_i) \geq \alpha \cdot \sum_{A \in \A} q_A v_i(A_j).
$$
When $\alpha = 1$, we call the lottery ex-ante envy-free. We use the phrase \emph{BoBW lottery with ex-post property $P$} to refer to a lottery that is ex-ante envy-free and satisfies $P$ ex post. 
For nonadditive valuations, the expectations are taken with respect to bundle values (and not those of individual items).

\paragraph{Envy-freeability.} An allocation $A$ is \emph{envy-freeable} if there is a nonnegative subsidy vector $p \in \R_{\geq 0}^n$ such that
$$
  v_i(A_i) + p_i \geq v_i(A_j) + p_j \qquad \text{for all } i,j \in N.
$$
In other words, the bundle-subsidy pairs are envy-free under quasilinear utilities $v_i(S) + p$. 

The \emph{envy graph} of allocation $A$ is the complete directed graph on $N$ with edge weights $w^v_A(i,j) \coloneqq v_i(A_j) - v_i(A_i)$.

\begin{theorem}[Characterisation of envy-freeability;~\citep{HS19fair,A21achieving}]
For every allocation $A$ and arbitrary real-valued valuations, the following are equivalent.
\begin{enumerate}[label=\textup{(\roman*)},leftmargin=*]
  \item $A$ is envy-freeable.
  \item The envy graph of $A$ has no positive-weight directed cycle.
  \item $A$ maximizes utilitarian welfare among all permutations of its bundles, i.e., for every permutation $\sigma$ of $N$, it holds that $\sum_{i \in N} v_i(A_i) \geq \sum_{i \in N} v_i(A_{\sigma(i)})$.
\end{enumerate}
When these conditions hold, the subsidies may be chosen as the maximum weights of directed paths starting at the corresponding vertices, including the zero-length path of weight $0$.
\label{thm:efable-char}
\end{theorem}

Envy-freeability is a form of welfare stability under reassignment of the existing bundles. However, it does not exclude welfare-improving reallocations that split the bundles. 
In \Cref{app: PO and EFable}, we show that Pareto optimality (PO) and envy-freeability are incomparable.

\paragraph{Scale-freeness.}
For a positive vector $c = (c_1, \dots, c_n) \in \Rpp^n$, the scaled profile $c \odot v$ is defined by $(c \odot v)_i(S) \coloneqq c_i v_i(S)$ for every $i \in N$ and $S \subseteq M$. An ex-post property $P$ is \emph{scale-free} if $\F_P(c \odot v) = \F_P(v)$ for every $v$ and $c\in\Rpp^n$. A valuation domain $\mathcal{V}$ is \emph{scale-free} if $c\odot v\in\mathcal{V}$ whenever $v\in\mathcal{V}$ and $c\in\Rpp^n$. 
Note that domain scale-freeness preserves admissible profiles, whereas property scale-freeness preserves admissible allocations.

The additive, submodular, subadditive, monotone, %
and lexicographic domains are scale-free. %
Matroid-rank valuations are not scale-free since scaling a unit marginal to $1/2$ requires leaving the domain. 

Many standard fairness and efficiency notions, such as exact envy-freeness and its approximations (EF, EF1, EFX), Pareto optimality (PO), fractional Pareto optimality (fPO; on additive profiles), and balancedness, are scale-free, as are proportionality and maximin-share.  By contrast, envy-freeability is not scale-free: Indeed, by condition (ii) of \Cref{thm:efable-char}, the net weight of an envy cycle may change sign upon scaling the valuation of one of the agents in the cycle. Likewise, equitability and its usual relaxations, such as EQ1 and EQX, also need not be scale-free because they compare utilities across agents.

\paragraph{Weighted notions.}
For fixed positive entitlements $w_1,\dots,w_n$, weighted envy-freeness compares $v_i(A_i)/w_i$ with $v_i(A_j)/w_j$. An allocation is \emph{weighted envy-freeable} if some $p \in  \R_{\geq 0}^n$ satisfies
$$
 \frac{v_i(A_i) + p_i}{w_i} \geq \frac{v_i(A_j) + p_j}{w_j} \qquad \text{for all } i,j \in N.
$$
For equal entitlements, weighted envy-freeability coincides with ordinary envy-freeability. Ex-ante weighted envy-freeness takes expectations of the same weighted bundle comparisons, without subsidies. For additive mixed manna, WEF1 eliminates each weighted envy comparison by deleting a positively valued item from the envied bundle or a negatively valued item from the envious bundle, with signs evaluated by the comparing agent. WEF1T permits transferring such an item between the two agents instead~\citep{LLTZ26weightedmixed}. Our main reduction in \Cref{sec:Reduction} concerns equal entitlements, and \Cref{thm:weighted} states the weighted extension.

\section{Reducing Randomization to Subsidies}\label{sec:Reduction}

In this section, we will establish our main reduction. Specifically, we will show that for any scale-free ex-post property $P$, if every profile in a scale-free domain admits an envy-freeable allocation satisfying $P$, then every profile in that domain admits a BoBW lottery with ex-post property $P$.

\begin{theorem}[BoBW-to-subsidy reduction]
Let $P$ be a scale-free ex-post property and fix a profile $v$. Suppose that, for every coordinatewise positive vector $c = (c_1, \dots, c_n) \in \Rpp^n$, the scaled profile $c \odot v$ admits an allocation that is $P$ and envy-freeable. Then, $v$ admits an ex-ante envy-free and ex-post $P$ lottery.
\label{thm:bobw-to-subsidy}
\end{theorem}

 The proof of \Cref{thm:bobw-to-subsidy} uses linear programming (LP) duality. The primal LP chooses a lottery that maximizes the smallest expected negative envy. Its dual assigns nonnegative weights to the directed pairwise comparisons. A negative value of the dual optimum certifies that every $P$-allocation violates the \emph{same} weighted comparison~(\Cref{lem:weighted-dual-separation}). We then show that a strictly positive vector of these weights can be decomposed into a product of agent-specific scaling factors and a \emph{circulation}~(\Cref{lem:decomposition}). The subsidy terms cancel out for such a circulation, giving the desired contradiction. In \Cref{cor:small-support}, we show that the lottery guaranteed by \Cref{thm:bobw-to-subsidy} can be chosen to have a support of size at most $n(n-1)+1$. \\ 

Let $E \coloneqq \{(i,j) \in N \times N : i \neq j\}$ and $K \coloneqq |E| = n(n-1)$. For each allocation $A$, define its \emph{negative envy vector} as $d(A) \coloneqq (d_{ij}(A))_{(i,j)\in E}$, where $d_{ij}(A) \coloneqq v_i(A_i)-v_i(A_j)$. Thus, $d_{ij}(A)=-w_A^v(i,j)$ is the negative of the weight of the corresponding edge in the envy graph of $A$. Observe that a lottery is ex-ante envy-free if and only if its expected negative envy vector is coordinatewise nonnegative. For any $\lambda \in \R^E$, we will write $\lambda \cdot d(A) = \sum_{e \in E} \lambda_e d_e(A)$.

\begin{lemma}[Weighted Dual Separation]\label{lem:weighted-dual-separation}
Suppose the profile $v$ admits an allocation satisfying $P$, i.e, $\F_P(v) \neq \emptyset$, but does not admit any lottery that is ex-ante envy-free and ex-post $P$. Then, there exists $\lambda \in \R_{>0}^E$ with $\sum_{e \in E} \lambda_e = 1$ such that $\lambda \cdot d(A) < 0$ for every $A \in \F_P(v)$.
\end{lemma}
\begin{proof}
The proof follows from LP duality. Consider any ex-post $P$ lottery for the profile $v$. Suppose the lottery associates a probability $q_A$ with an allocation $A \in \F_P(v)$ satisfying property $P$. Consider the following primal LP, whose variables are the lottery probabilities $q = (q_A)_{A \in \F_P(v)}$ and a
scalar $t \in \R$:
$$
\begin{aligned}
  \textup{maximize} \quad & t\\
  \textup{subject to} \quad
    & \sum_{A \in \F_P(v)} q_A d_e(A) \geq t && \text{ for all } e \in E,\\
    & \sum_{A \in \F_P(v)} q_A = 1,\\
    & q_A \geq 0 && \text{ for all } A \in \F_P(v).
\end{aligned}
$$
Its optimum $t^\star$ is the largest common lower bound on all expected negative envy values. In particular, a BoBW lottery exists if and only if $t^\star \geq 0$. The LP is feasible because $\F_P(v)$ is nonempty. Its optimum is finite and is attained, since, after eliminating $t$, the objective is the continuous function $q \mapsto \min_{e \in E} \sum_{A \in \F_P(v)} q_A d_e(A) $ on the compact probability simplex.

Let us associate a nonnegative dual variable $\lambda_e$ with each comparison constraint and a dual variable $z \in \R$ with the normalization constraint. The dual LP is given by
$$
\begin{aligned}
  \textup{minimize} \quad & z\\
  \textup{subject to} \quad
    & \lambda \cdot d(A) \leq z && \text{ for all } A \in \F_P(v),\\
    &\sum_{e \in E} \lambda_e = 1,\\
    &\lambda_e \geq 0 && \text{ for all } e \in E.
\end{aligned}
$$
The dual can be observed to be feasible by instantiating it for $\lambda_e = 1/K$ for all $e \in E$ and $z = \max_{A \in \F_P(v)} \sum_{e \in E} d_e(A)/K$. Strong LP duality therefore gives an optimal dual solution $(\lambda^0,z^\star)$ with $z^\star=t^\star$.

If no BoBW lottery exists, then $z^\star = t^\star < 0$. Let $\delta \coloneqq -z^\star>0$. Then, every allocation $A \in \F_P(v)$ satisfies $\lambda^0 \cdot d(A) \leq - \delta$. 

Since $\lambda^0$ satisfies the dual, it is coordinatewise nonnegative. We will now perturb $\lambda^0$ to construct the desired coordinatewise strictly positive weight vector. Let $u_e \coloneqq 1/K$ for all $e \in E$, and set
$$
 B \coloneqq \max_{A \in \F_P(v)} |u \cdot d(A)|,
 \qquad \eps = \frac{\delta}{2(\delta + B)}.
$$
Consider the vector $\lambda \coloneqq (1 - \eps) \lambda^0 + \eps u$. Observe that all coordinates of $\lambda$ are strictly positive and sum to $1$. Consequently, for every allocation $A \in \F_P(v)$, we have
$$
 \lambda \cdot d(A) \leq -(1 - \eps) \delta + \eps B = -\frac{\delta}{2} < 0,
$$
as desired.
\end{proof}

The next decomposition lemma~(\Cref{lem:decomposition}) states that any coordinatewise positive vector $\lambda \in \R_{>0}^E$ can be expressed as an agent-specific scaling of a \emph{balanced flow} or \emph{circulation}. A nonnegative vector $\mu=(\mu_{ij})_{(i,j)\in E}$ is a \emph{circulation} if, for every agent, the total outgoing and incoming weights are equal, i.e., $\sum_{j \neq i} \mu_{ij} = \sum_{j \neq i} \mu_{ji}$ for every $i \in N$. As we will see later in the proof of \Cref{thm:bobw-to-subsidy}, this balance property will lead to the cancellation of subsidy terms.

\begin{restatable}[Decomposition into scaling and circulation]{lemma}{DecompositionLemma}
For every $\lambda \in \R_{>0}^E$, there is a coordinatewise positive scaling vector $c \in \Rpp^n$ and a strictly positive circulation $\mu \in \R_{>0}^E$ such that $\lambda_{ij} = c_i \mu_{ij}$ for every $(i,j)\in E$. Moreover, $\mu$ can be chosen with $\sum_{(i,j) \in E} \mu_{ij} = 1$.
\label{lem:decomposition}
\end{restatable}

This lemma is a direct application of the standard stationary-distribution theorem for finite irreducible Markov chains; see, for example, Section~2.1.1 of \citet{AF02reversible}. The circulation guaranteed by the lemma is the chain's stationary edge flow. We include a short proof in \Cref{appendix:redn} to keep the discussion self-contained.

Let us now use \Cref{lem:weighted-dual-separation,lem:decomposition} to prove \Cref{thm:bobw-to-subsidy}.

\begin{proof}[Proof of \Cref{thm:bobw-to-subsidy}]

First, observe that $\F_P(v)$ is nonempty by considering the all-ones vector $c = \one =(1,\dots,1) \in \Rpp^n$ in the hypothesis of the theorem. 

Suppose, for contradiction, that the profile $v$ has no BoBW lottery with ex-post property $P$. By \Cref{lem:weighted-dual-separation}, there exists $\lambda \in \R_{>0}^E$ with $\sum_{e \in E} \lambda_e = 1$ such that $\lambda \cdot d(A) < 0$ for every $A \in \F_P(v)$. From \Cref{lem:decomposition}, there exists a strictly positive vector $c \in \Rpp^n$ and a strictly positive circulation $\mu \in \R_{>0}^E$ such that $\lambda_{ij} = c_i \mu_{ij}$ for every $(i,j)\in E$.

By the theorem's hypothesis, the scaled profile $c \odot v$ has an allocation $A$ that is envy-freeable and satisfies $P$. That is, there exist nonnegative subsidies $p = (p_1, \dots, p_n)$ such that
\begin{equation}
    c_i v_i(A_i) + p_i \geq c_i v_i(A_j) + p_j \qquad\text{for all } i,j \in N.
\label{eqn:reduction_helper-1}
\end{equation}
Since $P$ is scale-free, $A$ also satisfies $P$ with respect to $v$, i.e., $A \in \F_P(v)$. Moreover, from \Cref{eqn:reduction_helper-1}, we get $c_i d_{ij}(A) \geq p_j - p_i$. Multiplying these inequalities by $\mu_{ij}$ and summing over $(i,j) \in E$:
$$
\begin{aligned}
  \lambda \cdot d(A)
    & = \sum_{(i,j) \in E} \mu_{ij} c_i d_{ij}(A)\\
    & \geq \sum_{(i,j) \in E} \mu_{ij}(p_j - p_i)\\
    & = \sum_{k \in N} p_k \left(\sum_{i \neq k} \mu_{ik} - \sum_{j \neq k} \mu_{kj} \right)\\
    & = 0.
\end{aligned}
$$
The last equality follows from the balanced flow property of the circulation $\mu$. This contradicts the dual separation condition $\lambda \cdot d(A) < 0$ derived earlier. Therefore, the primal optimum is nonnegative, and an ex-ante envy-free and ex-post $P$ lottery exists.
\end{proof}

\begin{corollary}[Small support]
For any valuation profile $v$, any ex-post property $P$ and any lottery over $\F_P(v)$, there exists another lottery over $\F_P(v)$ with the same expected negative envy vector and support of size at most $n(n-1) + 1$. In particular, whenever an ex-ante envy-free and ex-post $P$ lottery exists, there also exists one with this support bound.
\label{cor:small-support}
\end{corollary}

\begin{proof}
The argument follows from Carath\'eodory's theorem applied to the negative envy vectors in $\R^E$. 

Fix a lottery over $\F_P(v)$, and choose, among lotteries supported on its support with the same expected negative envy vector, one with minimum support size $s$. Denote its support as $A^1, \dots, A^s$ and their positive probabilities as $q_1, \dots, q_s$. 

If $s > K + 1$, then the augmented vectors $(d(A^t), 1)$ in $\R^{K+1}$ are linearly dependent. Thus, some nonzero $(a_1, \dots, a_s)$ satisfies $\sum_{t = 1}^s a_t d(A^t) = 0$ and $\sum_{t = 1}^s a_t = 0$. In particular, some $a_t$ is positive. 

Let $\tau \coloneqq \min_{t: a_t > 0} q_t/a_t$ and $q'_t \coloneqq q_t - \tau a_t$. The vector $q' \coloneqq (q'_1, \dots, q'_s)$ can be interpreted as a probability vector since its entries are nonnegative and add up to $1$. Furthermore, $q'$ preserves the expected negative envy vector, and at least one of its coordinates is zero without any previously-zero entry becoming positive. This contradicts the minimality of $s$. Hence, $s \leq K + 1 = n(n-1) + 1$. 
\end{proof}

\begin{corollary}[Reduction on a scale-free domain]
Let $\mathcal{V}$ be a scale-free valuation domain and let $P$ be a scale-free ex-post property. If every profile in $\mathcal{V}$ admits an allocation that is both $P$ and envy-freeable, then every profile in $\mathcal{V}$ admits an ex-ante envy-free and ex-post $P$ lottery supported on at most $n(n-1)+1$ allocations.
\label{cor:domain-reduction}
\end{corollary}

\begin{proof}
Fix a profile $v \in \mathcal{V}$. For every $c \in \Rpp^n$, scale-freeness of the domain gives $c \odot v \in \mathcal{V}$. Therefore, each scaled profile $c \odot v$ admits some allocation that is $P$ and envy-freeable. From \Cref{thm:bobw-to-subsidy}, we obtain the existence of a lottery satisfying ex-ante envy-freeness and ex-post $P$ with respect to $v$. The support size bound follows from \Cref{cor:small-support}.
\end{proof}

Note that our analysis does not use additivity %
or any sign restriction on the valuations, and therefore applies to goods, chores, and mixed instances, with the appropriate choice of $P$.

\section{Consequences of Our Reduction}
\label{sec:consequences}

\subsection{Recovering additive BoBW existence}
A useful consequence of our main reduction (\Cref{thm:bobw-to-subsidy}) is that we obtain alternative existence proofs of BoBW lotteries for additive goods, additive chores, and additive mixtures of goods and chores.

\begin{corollary}[Additive BoBW]
Every additive goods-only or chores-only instance admits a lottery that is ex-ante EF and ex-post balanced and EF1. Furthermore, every additive instance with mixed goods and chores admits a lottery that is ex-ante EF and ex-post EF1.
\label{cor:additive-bobw}
\end{corollary}
\begin{proof}
\citet{BDN+20one} provided an algorithm to compute an allocation that is balanced, EF1, and envy-freeable for additive goods. Later,~\citet{WZZ23} observed that this algorithm can be adapted for additive chores. Balancedness and EF1 are scale-free, and therefore, it follows from \Cref{thm:bobw-to-subsidy} that every additive goods (respectively, chores) instance admits a lottery that is ex-ante EF and ex-post balanced and EF1. This recovers the existential guarantee of~\citet{AFS+24best} for goods-only and chores-only problems.

For the case of mixtures of goods and chores,~\citet{ALMS25mixed} showed the existence of allocations that are simultaneously EF1 and envy-freeable. Using \Cref{thm:bobw-to-subsidy}, we recover the corresponding BoBW result of~\citet{ABLMSTW26mixed}.
\end{proof}

\subsection{A converse for two agents}

\begin{theorem}[Two-agent converse]
Let $n=2$, and let $q$ be an ex-ante EF lottery. Then, for every $c \in \R_{>0}^2$, some $A \in \supp(q)$ is envy-freeable at $c \odot v$. 
Furthermore, for a scale-free property $P$, a lottery that is ex-ante EF and ex-post $P$ exists if and only if every positive rescaling admits an allocation that is simultaneously $P$ and envy-freeable.
\label{thm:two-agent}
\end{theorem}
\begin{proof}
Ex-ante envy-freeness of the lottery $q$ implies that $\mathbb{E}_q[c_1 d_{12}(A) + c_2 d_{21}(A)] \geq 0$. Therefore, some allocation $A$ in the support satisfies $c_1 d_{12}(A) + c_2 d_{21}(A) \geq 0$. This condition implies that the only directed envy cycle has nonpositive weight in the scaled profile. From \Cref{thm:efable-char}, it follows that $A$ is envy-freeable.

For the second part of the claim, the `if' direction follows from \Cref{thm:bobw-to-subsidy}. The `only if' direction follows from the observation that if the support satisfies a scale-free property $P$, then it continues to do so after rescaling.
\end{proof}

It is relevant to note that the order of quantifiers in \Cref{thm:two-agent} cannot be switched. Specifically, the claim holds for $\forall c \, \exists A$, and not for $\exists A \, \forall c$. Indeed, consider an instance with two agents who each value one good at $1$. Allocating the good uniformly to either agent is ex-ante EF. Assigning the item deterministically to agent $1$ is envy-freeable only when $c_1 \geq c_2$, while assigning the item to agent $2$ is envy-freeable only when $c_2 \geq c_1$.

\begin{corollary}\label{cor:two-agent corollary}
Any two-agent instance with additive goods admits an allocation that is EFX$_{+}$, PO, and envy-freeable.
\end{corollary}
\begin{proof}
For two agents,~\citet{Aziz26po} showed the existence of an ex-ante EF and ex-post EFX$_{+}$ and PO lottery. \Cref{thm:two-agent} implies that every such instance admits an allocation that is EFX$_{+}$, PO and envy-freeable.
\end{proof}

\subsection{Weighted entitlements}

We will now show that our general reduction extends to fair division with unequal entitlements. Towards this goal, we will use the characterization of weighted envy-freeability due to~\citet{AHKS+25Subsidy}. It uses the notion of \emph{weighted envy graph} of allocation $A$, which is the complete directed graph on $N$ with the weight of edge $(i,j)$ given by $v_i(A_j)/w_j - v_i(A_i)/w_i$.

\begin{theorem}[Characterisation of weighted envy-freeability;~\citep{AHKS+25Subsidy}]
For every allocation $A$ and arbitrary real-valued valuations, the following are equivalent.
\begin{enumerate}[label=\textup{(\roman*)},leftmargin=*]
  \item $A$ is weighted envy-freeable.
  \item The weighted envy graph of $A$ has no positive-weight directed cycle.
  \item $A$ satisfies weighted reassignment stability, i.e., for every permutation $\sigma$ of $N$, it holds that 
  $$\sum_{i \in N} \frac{v_i(A_i)}{w_i} \geq \sum_{i \in N} \frac{v_i(A_{\sigma(i)})}{w_{\sigma(i)}}.$$
\end{enumerate}
\label{thm:efable-char-weighted}
\end{theorem}

We will now state and prove the weighted extension of our reduction.

\begin{theorem}[Weighted reduction]
Let $(w_1,\dots,w_n)$ denote the positive entitlements of the $n$ agents. Let $P$ be invariant under positive agent-specific utility rescaling. If every positive rescaling of $v$ admits a $P$-allocation that is weighted envy-freeable, then $v$ admits an ex-ante WEF lottery supported on $P$-allocations. Furthermore, for two agents, the corresponding converse statement holds.
\label{thm:weighted}
\end{theorem}
\begin{proof}
Define $d^w_{ij}(A) = v_i(A_i)/w_i - v_i(A_j)/w_j$. The LP and its dual in \Cref{lem:weighted-dual-separation} apply to these weighted negative envy vectors. Positive rescaling multiplies row $i$ by $c_i$. Given a decomposition $\lambda_{ij} = c_i \mu_{ij}$ and subsidies for the scaled profile, let $s_i \coloneqq p_i/w_i$. The weighted subsidy inequalities imply $c_id^w_{ij}(A) \geq s_j - s_i$. Summing the inequalities edgewise over the circulation gives $\lambda \cdot d^w(A) \geq 0$, contradicting the dual separation condition. This proves the BoBW-to-subsidy reduction for the weighted setting.

For two agents, fix any $c \in \mathbb{R}^2_{>0}$ and any ex-ante WEF lottery $q$. As in~\Cref{thm:two-agent}, we have $\mathbb{E}_q[c_1 d^w_{12}(A) + c_2 d^w_{21}(A)] \geq 0$. This implies that some allocation $A$ in the support satisfies $c_1 d^w_{12}(A) + c_2 d^w_{21}(A) \geq 0$, which implies that any directed cycle in the weighted envy graph of $A$ at the scaled profile $c \odot v$ has a nonpositive weight. From \Cref{thm:efable-char-weighted}, it follows that allocation $A$ is weighted envy-freeable at that scale.
\end{proof}

\begin{corollary}
For two agents with fixed positive entitlements and additive mixed-manna valuations, a WEF1T and weighted envy-freeable allocation exists.
\label{cor:weighted-mixed}
\end{corollary}
\begin{proof}
For additive mixed manna,~\citet{LLTZ26weightedmixed} proved the existence of a lottery that is ex-ante weighted envy-free and ex-post WEF1T. Applying \Cref{thm:weighted}, every two-agent additive instance with mixed goods and chores admits an allocation that is simultaneously WEF1T and weighted envy-freeable.
\end{proof}

\subsection{Three-agent lexicographic goods}

For additive valuations,~\citet{Aziz26po} showed that an ex-ante EF and ex-post EFX and PO lottery can fail to exist for three agents. In this section, we will show that such a lottery always exists for three agents with lexicographic valuations.  

\begin{theorem}[Three-agent lexicographic existence]
Any goods instance with three agents and lexicographic valuations admits an allocation that is simultaneously EFX, PO, and envy-freeable. Consequently, an ex-ante EF lottery over EFX+PO allocations exists.
\label{thm:lex-three}
\end{theorem}
\begin{proof}
For lexicographic valuations, it is known that an allocation $A$ is EFX if and only if every envied agent receives exactly one good~\citep{HSV+21fair}. Moreover, allocation $A$ is PO if and only if there is a picking sequence of agents that generates that allocation. We use these two characterizations in our proof.

We assume $|M|>3$ throughout. Indeed, when $|M|\leq3$, we can add dummy goods to make $|M|=3$ exactly. Let $A$ be a maximum weight matching between agents and goods. From \Cref{thm:efable-char}, $A$ is envy-freeable. It is also EFX as every agent has at most one good. Furthermore, if $A$ is Pareto dominated, then every agent must get a good it strictly prefers over its assignment in $A$. This would contradict welfare maximality of $A$, implying it is also PO. 

Let $N=\{1,2,3\}$ denote the set of agents. 
Let $F$ denote the set of favorite goods of the agents, and let $k \coloneqq |F| \geq 1$ denote the number of \emph{distinct} favorite goods of the agents. We prove the claim via case analysis on $k$.

    First, suppose $k=3$. In this case, we give each agent $i \in [3]$ its favorite good and all remaining goods to some arbitrary agent (say, agent 3). This results in an envy-free allocation and hence, is EFX and envy-freeable. The allocation is clearly PO as it can be simulated by the picking sequence $(1,2,3,\dots,3)$.
    
    Suppose $k=2$. Let us relabel the agents such that agents $1$ and $2$ have the common favorite good $a$ and agent 3's favorite good is $b \neq a$. Let $R= M\setminus \{a,b\}$. Let $A$ be the utilitarian welfare-maximizing allocation (equivalently, a permutation) of the bundles $\{a\},\{b\},$ and $R$. By \Cref{thm:efable-char}, $A$ is envy-freeable. We claim that $A$ is also EFX and PO.
    Note that agent $3$ cannot receive good $a$, as swapping $a$ and $b$ would strictly improve both agents. Let $i \in \{1,2\}$ receive $a$ and let $j$ be the other agent in that pair. If agent $3$ receives $b$, then the allocation is sequencible (hence PO) and EFX. Otherwise, agent $j$ receives $b$ and agent $3$ receives $R$. Welfare maximality gives $v_j(b) - v_j(R) \geq v_3(b) - v_3(R) > 0$. Thus, $b$ is agent $j$'s favorite good outside of $a$ and therefore the allocation is again sequencible and EFX.

        Let us now consider the case $k=1$. Let $a$ denote the common favorite good, and let $b_i$ be agent $i$'s favorite good in $M\setminus\{a\}$. 
        
        Consider the case when not all $b_i$'s are distinct. We relabel the agents such that $b_1=b_2=b$. Let $A$ be the utilitarian welfare-maximizing allocation of the bundles $\{a\},\{b\},R$ where $R = M\setminus \{a,b\}$. As before, $A$ is envy-freeable. Let $h_a,h_b,h_R$ denote the agent assigned $\{a\},\{b\},$ and $R$ in $A$, respectively. If $A$ is not EFX, then $h_b$ envies $R$. This implies $h_b$ is agent $3$ since agents $1$ and $2$ prefer $b$ to $R$. But then, as $b_3 \in R$, swapping $R$ and $b$ strictly increases welfare, contradicting the optimality of $A$. Hence, $A$ is EFX. Also note that $A$ is sequencible (hence PO) as it can be simulated via the picking sequence $h_a, h_b, h_R,\dots, h_R.$
        
        If all three $b_i$'s are distinct, choose $s \in \arg\min_{i\in[3]} \{v_i(a)-v_i(b_i)\}$. Let $A$ be the welfare-maximizing allocation of the bundles $\{a\},\{b_s\},\{R\}$. Suppose $A$ is not EFX. This forces $s\neq h_{b_s}$ and $h_{b_s}$ prefers $R$ to $\{b_s\}$. If $h_R = s$, swapping $R$ and $\{b_s\}$ strictly improves the total welfare. If $h_a = s$, then we can swap bundles $\{b_s\}$ and $\{a\}$. The total welfare strictly improves by our choice of $s$. Hence, both scenarios contradict the optimality of $A$. Hence, $A$ is EFX. To show PO, observe that $A$ can be simulated via the picking sequence $h_a, s, h_R,\dots , h_R$.

Thus, an allocation that is envy-freeable, EFX, and PO always exists for $n=3$ agents. From~\Cref{thm:bobw-to-subsidy}, it follows that ex-ante EF and ex-post EFX+PO lottery always exists for three lexicographic agents.
\end{proof}

\begin{corollary}
Let $P$ be a scale-free property. If profile $v$ admits no ex-ante EF and ex-post $P$ lottery, then there exists some $c\in\R_{>0}^n$ for which no $P$-allocation is envy-freeable at $c\odot v$.
\label{cor:impossibility}
\end{corollary}

The above corollary is the contrapositive of \Cref{thm:bobw-to-subsidy}. We next present two impossibilities that follow as direct consequences. 

\begin{corollary}\label{impossible:EF1+fPO+Efable}
There exists an additive goods instance where no allocation simultaneously satisfies EF1, fPO, and envy-freeability.

\end{corollary}
\begin{proof}
\citet{AFS+24best} provided an additive goods instance where no lottery is ex-ante EF and ex-post EF1 and fPO. Both EF1 and fPO are scale-free for the domain of additive valuations. Therefore, by \Cref{cor:impossibility}, there exists some scaling of the given instance where no allocation is simultaneously EF1, fPO and envy-freeable.     
\end{proof}
\vspace{-0.2cm}
\begin{corollary}
There exists an additive chores instance such that no allocation is WEF1 and weighted envy-freeable.
\label{impossible: WEF1 and EFable chores}
\end{corollary}
\begin{proof}
\citet{LWZ26binary} showed that ex-ante weighted envy-freeness and ex-post WEF1 can be incompatible even for two agents with identical binary chore costs. \Cref{cor:impossibility} gives us the required result. 
\end{proof}
\begin{corollary}
There exists an additive goods instance such that no allocation is WEF1 and weighted envy-freeable.
\label{impossible: WEF1 and EFable goods}
\end{corollary}

\begin{proof}
\citet{HSV24entitlements} and~\citet{AGM23BoBW} showed that ex-ante weighted envy-freeness and ex-post WEF1 can be incompatible. The required result then follows from~\Cref{cor:impossibility}. This recovers the non-existence guarantee of \citet{AHKS+25Subsidy}. 
\end{proof}

\section{Limits of the Reduction}
\label{sec:Limits}

\subsection{A submodular example with three agents}

\begin{theorem}[Submodular separation] %
There is an instance with four goods and three agents having monotone submodular valuations such that no allocation is simultaneously EF1 and envy-freeable, yet an ex-ante EF and ex-post EF1 lottery exists.
\label{thm:counterexample-submodular-valuations}
\end{theorem}

\begin{proof} 
Let $M = \{a,b,c,d\}$ denote the set of goods, and let $t$ and $u$ be additive functions with item values of $(a,b,c,d)$ given by $t = (5,8,11,0)$ and $u = (9,1,0,10)$. 

Define agent $1$'s valuation function as $v_1(S) \coloneqq \min\{11,t(S)\}+\min\{10,u(S)\}$. Construct $v_2(S)$ and $v_3(S)$ by cyclically rotating $a,b,c$ while keeping $d$ fixed. That is, if $\rho$ denotes $a\to b\to c\to a$, then $v_i(S) \coloneqq v_1(\rho^{1-i}(S))$ for $i \in \{2,3\}$. Let us write $ac$ to denote $\{a,c\}$. Then, the relevant values are
\begin{center}
\begin{tabular}{c@{\qquad}rrrr@{\qquad}rr}
 & $a$ & $b$ & $c$ & $d$ & $ac$ & $ad$ \\
\midrule
$v_1$ & 14 & 9 & 11 & 10 & 20 & 15 \\
$v_2$ & 11 & 14 & 9 & 10 & 12 & 21 \\
$v_3$ & 9 & 11 & 14 & 10 & 21 & 18 \\
\end{tabular}
\end{center}

Observe that each $v_i$ is the sum of two budget-additive functions, and is therefore submodular \citep{GHMbudget}.

\emph{No EF1+envy-freeable allocation.} All singleton values in our construction are positive. Therefore, any EF1 allocation must have bundle sizes $(2,1,1)$ up to reordering. The pair bundle in any EF1 partition can contain $d$ or not. By cyclic symmetry, we only need to consider the two unordered partitions $\{ac,b,d\}$ and $\{ad,b,c\}$.

From~\Cref{thm:efable-char}, we know that any envy-freeable allocation maximizes welfare among permutations of its bundles. It can be checked that for the partition $\{ac,b,d\}$, the unique welfare-maximizing permutation is $(d,b,ac)$, while that for $\{ad,b,c\}$ is $(b,ad,c)$. Neither of these allocations is EF1, giving us the claimed impossibility.

\emph{A BoBW lottery.} Consider a uniform distribution over the following allocations: $(ac,b,d)$, $(d,ab,c)$, and $(a,d,bc)$. These allocations are cyclic rotations of one another. Additionally, it can be easily checked that the first one is EF1, and therefore, by cyclic symmetry, all are. The expected negative envy matrix is
$$\left(\mathbb E[v_i(A_i)-v_i(A_j)]\right)_{i,j = 1}^3 = \frac{1}{3}
 \begin{pmatrix}
 0 & 4 & 11\\
 11 & 0 & 4\\
 4 & 11 & 0
 \end{pmatrix}.
$$
Since every off-diagonal entry is positive, ex-ante envy-freeness follows.
\end{proof}

\subsection{A lexicographic example with four agents}

We now present an example with lexicographic valuations where the converse of the reduction does not hold.

\begin{theorem}[EFX + Envy-freeability counterexample] \label{thm:counterexample-lexicographic-four}
There is an instance with five goods and four agents with lexicographic valuations such that no allocation is simultaneously EFX and envy-freeable, yet there exists an ex-ante EF and ex-post EFX+PO lottery.
\end{theorem}
\begin{proof}
Consider the following instance with four agents $N=\{1,2,3,4\}$ and five goods $M=\{a,b,c,d,e\}$:
\begin{table}[h]
\centering
    \begin{tabular}{c|ccccc}
        & $a$ & $b$ & $c$ & $d$ & $e$ \\  
        \cmidrule{1-6}
        $1$ & $26$ & $1$ & $8$ & $2$ & $12$\\
        $2$ & $16$ & $8$ & $2$ & $4$ & $1$ \\
        $3$ & $8$ & $16$ & $4$ & $1$ & $2$ \\
        $4$ & $2$ & $28$ & $1$ & $12$ & $8$
    \end{tabular}
\end{table}

Observe that $a$ is the favorite good of both agents $1$ and $2$. Thus, in any allocation, the bundle containing $a$ is envied by at least one of these two agents. The same holds for good $b$ and agents $3$ and $4$. For lexicographic valuations, an allocation is EFX if and only if every envied agent receives exactly one good~\citep{HSV+21fair}. All item values are positive, so an EFX allocation has bundle sizes $(2,1,1,1)$. The goods $a$ and $b$ must be singleton bundles, since each is the favorite of two agents. Thus, the underlying partition of any EFX allocation is one of the following:
    
\begin{itemize}[leftmargin=*]
    \item $\{a\},\{b\},\{c\}, \{d,e\}$: The allocation $A^{(c)}=\left(\{a\},\{c\},\{b\},\{d,e\}\right)$ is the unique permutation of the bundles that maximizes total welfare. However, $A^{(c)}$ is not EFX as agent $2$ envies agent $4$ after removing $e$. 
        
    \item $\{a\},\{b\},\{d\}, \{c,e\}$: The allocation $A^{(d)}=(\{c,e\},\{a\},\{d\},\{b\})$ is the unique permutation of the bundles maximizing welfare. This allocation is not EFX as agent $3$ envies agent $1$ after removing $e$.
    
    \item $\{a\},\{b\},\{e\}, \{c,d\}$: The allocation $A^{(e)}=(\{a\},\{c,d\},\{e\},\{b\})$ is the unique permutation of the bundles maximizing utilitarian welfare. Once again, $A^{(e)}$ is not EFX as agent $3$ envies agent $2$ after removing $d$.

\end{itemize}

Hence, no allocation is simultaneously envy-freeable and EFX. 

We will now show that the following lottery is ex-ante EF and ex-post EFX and PO:
\begin{align*}
X &= \frac{1}{2}\left(\{a\},\{d,e\},\{c\},\{b\}\right)+\frac{1}{2}\left(\{c,e\},\{a\},\{b\},\{d\}\right) .    
\end{align*}

The two allocations in the support of $X$ are induced by the picking sequences $(1,4,3,2)$ and $(2,3,4,1)$, and are therefore Pareto optimal~\citep{HSV+21fair}. Additionally, since each envied agent has exactly one good, these are also EFX. 
Hence, the lottery $X$ is ex-post EFX+PO. The expected negative envy matrix is:
$$
    \left(\mathbb E[v_i(A_i)-v_i(A_j)]\right)_{i,j=1}^4 = \frac{1}{2}\begin{pmatrix}
    0 & 6 & 37 & 43\\
    2 & 0 & 11 & 9\\
    6 & 9 & 0 & 3\\
    29 & 18 & 11 & 0
\end{pmatrix}.
$$
Since every off-diagonal entry is positive, $X$ is ex-ante envy free.
\end{proof}

\section{Concluding Remarks}

We established a general connection between fairness via subsidies and fairness through randomization. Our result unifies several existing observations and contributes new ones. The converse holds for two agents but fails for three agents, as our example for submodular valuations demonstrate. We also provide the first incompatibility result between EF1 and envy-freeability beyond the additive setting, resolving an open question of~\citet{ALMS25mixed}.

Our work opens up several directions for future work. First, it would be useful to identify \emph{additional scale-free properties and domains} in which positive subsidy results and negative BoBW results can be used to infer the corresponding implications for the other problem. 
Second, it would be useful to develop \emph{computational analogs} of our reduction in which algorithms for the subsidy problem can be transformed into algorithms for BoBW. Finally, it would be interesting to explore if there are any \emph{quantitative} extensions of our result, wherein one can use bounds on subsidy to obtain exact or approximate BoBW guarantees.

\section*{Acknowledgments}
RV acknowledges support from DST INSPIRE grant no. DST/INSPIRE/04/2020/000107, ANRF grant no. CRG/2022/002621, and Mr. D.P. Gupta Chair Professorship.

\bibliographystyle{plainnat} 
\bibliography{References-2}

\appendix
\clearpage
\section{Omitted Material from Section~\ref{sec:Reduction}}
\label{appendix:redn}

\DecompositionLemma*

\begin{proof}
For each $i \in N$, let $r_i \coloneqq \sum_{j \neq i} \lambda_{ij} > 0$. Define an $n\times n$ matrix $Q$ as $Q_{ij} \coloneqq \lambda_{ij} / r_i$ for all $i \neq j$ and $Q_{ii} = 0$ for all $i$. 

The matrix $Q$ is row-stochastic (i.e., every row of $Q$ sums up to one) and every off-diagonal entry is positive, so it defines the transition matrix of a finite irreducible Markov chain. 

We first
give an elementary proof that $Q$ has a coordinatewise positive
stationary probability row vector $\pi$. Consider a probability row vector $\rho = (1/n,\dots,1/n)$ and, for each integer $T \geq 1$, define
$$
  \pi^{(T)} \coloneqq \frac{1}{T} \sum_{t=0}^{T-1} \rho Q^t.
$$
Each $\pi^{(T)}$ is a probability vector. By compactness of the probability simplex, a subsequence converges to a probability vector $\pi$. Since
$$
  \pi^{(T)} Q - \pi^{(T)} = \frac{\rho Q^T - \rho}{T} \longrightarrow 0,
$$
we have $\pi Q = \pi$ and $\sum_i \pi_i = 1$. Moreover, every $\pi_j$ is positive. Indeed, if $\pi_j = 0$, some $i \neq j$ has $\pi_i > 0$, and then $\pi_j = \sum_k \pi_k Q_{kj} \geq \pi_i Q_{ij} > 0$, a contradiction.

Now set $\mu_{ij} = \pi_i Q_{ij}$, and $c_i = \frac{r_i}{\pi_i}$. All these quantities are positive. The row sums of $Q$ and stationarity give $\sum_{j \neq i} \mu_{ij} = \pi_i = \sum_{j \neq i} \mu_{ji}$, so $\mu$ is a circulation and $\sum_{(i,j) \in E} \mu_{ij} = \sum_i \pi_i = 1$. Finally, $c_i \mu_{ij} = r_i Q_{ij} = \lambda_{ij}$, as required.
\end{proof}

\section{PO and envy-freeability}
\label{app: PO and EFable}

In this section, we show that PO and envy-freeability are incomparable, as neither notion implies the other. Let $\varepsilon >0$ be sufficiently small. Consider the following instance with two goods and two agents:
\begin{table}[H]
    \centering
    \begin{tabular}
{c@{\qquad}rrrr@{\qquad}rr}
        & $a$ & $b$ \\  
        \midrule{}
        $1$ & $1$ & $2+\varepsilon$ \\
        $2$ & $0$ & $1$
    \end{tabular}
\end{table}

Consider the singleton bundles $\{a\}$ and $\{b\}$. The allocation $(b,a)$ has utilitarian welfare $2+\varepsilon$ while allocation $(a,b)$ has welfare 2. Using \Cref{thm:efable-char}, we note that allocation $(b,a)$ is envy-freeable. However, it is not PO as the allocation $(\{b,a\},\emptyset)$ dominates it. On the other hand, note that while $(a,b)$ is not envy-freeable, it is still Pareto optimal.

\end{document}